\documentclass[a4paper,11pt,onecolumn,unpublished]{quantumarticle}
\pdfoutput=1
\usepackage[T1]{fontenc}
\usepackage[utf8]{inputenc}
\usepackage{amsmath,amssymb,amsthm,mathtools}
\usepackage{booktabs,enumitem,microtype,tikz}
\usepackage{hyperref}
\hypersetup{colorlinks=true,linkcolor=blue,citecolor=blue,urlcolor=blue}
\newtheorem{theorem}{Theorem}[section]
\newtheorem{proposition}[theorem]{Proposition}
\newtheorem{lemma}[theorem]{Lemma}
\newtheorem{corollary}[theorem]{Corollary}
\newtheorem{conjecture}[theorem]{Conjecture}
\theoremstyle{definition}
\newtheorem{definition}[theorem]{Definition}
\newtheorem{remark}[theorem]{Remark}
\newcommand{\Tr}{\operatorname{Tr}}

\newcommand{\id}{\mathbb I}
\newcommand{\R}{\mathbb R}
\newcommand{\1}{\mathbf 1}
\newcommand{\Coeff}{\mathcal C}
\newcommand{\Creg}{\mathfrak C}
\newcommand{\Qreg}{\mathfrak Q}
\newcommand{\Greg}{\mathfrak G}
\newcommand{\Sopt}{\mathsf S}
\newcommand{\btr}{\beta_{\mathrm{tr}}}
\newcommand{\doi}[1]{\href{https://doi.org/#1}{\nolinkurl{#1}}}
\setlist[enumerate]{itemsep=3pt,topsep=4pt}
\title{Classical, quantum, and general probabilistic state-discrimination profiles}
\author{Mih\'aly Weiner}
\affiliation{Institute of Mathematics,
	Budapest University of Technology and Economics, Műegyetem rkp.\ 3.,
	H-1111 Budapest, Hungary}

\begin{document}
\begingroup
\makeatletter
\patchcmd{\@titleatfontsize}{\@title}{\@title\par}{}%
  {\ClassWarning{quantumarticle}{Title paragraph adjustment not applied}}%
\makeatother
\exhyphenpenalty=10000 
\maketitle
\endgroup
\begin{abstract}
For a collection of states indexed by $[n]=\{1,\ldots,n\}$, let $F(H)$ denote
the minimum unnormalized discrimination error for the subfamily indexed by
each nonempty $H\subseteq[n]$. We call $F$ the \emph{discrimination profile} of
the collection and characterize which profiles can arise in a general
probabilistic theory (GPT). For fixed $n$, the set $\Greg_n$ of all GPT profiles
turns out to be a bounded rational polytope, giving a discrimination-theoretic
hierarchy
\[
 \Creg_n\subseteq\Qreg_n\subseteq\Greg_n,
\]
where $\Creg_n$ is the classical profile polytope (which we also characterize)
and $\Qreg_n$ is the convex set of quantum profiles. The hierarchy mirrors the
local--quantum--no-signalling hierarchy in Bell theory.

For $n=3$ -- the smallest case exhibiting a separation -- we give an explicit
description of the classical and GPT profile polytopes. Within $\Greg_3$, the
classical region is cut out by the single inequality
\[
 \beta:=F(123)-F(12)-F(13)-F(23)\le\beta_{\mathrm C}\equiv0,
\]
whose GPT maximum is $\beta_{\mathrm{GPT}}=1$. For qubit triples, the maximal
$\beta$, attained at the trine, is $\btr=3\sqrt3/2-2$, already violating the
classical bound; we further give an elementary proof that the
dimension-unrestricted quantum supremum satisfies $\beta_{\mathrm Q}\le2/3$.
Consequently, the single functional $\beta$ already separates the three profile
regions strictly:
\[
 0=\beta_{\mathrm C}
 <\frac{3\sqrt3}{2}-2
 \le\beta_{\mathrm Q}
 \le\frac23
 <1=\beta_{\mathrm{GPT}}.
\]

The bound $\beta_{\mathrm Q}\le2/3$ is not tight; we give a further argument
yielding a slight improvement, which however likely remains suboptimal. We
conjecture $\beta_{\mathrm Q}=\btr=3\sqrt3/2-2$, and support this by showing that
the trine is optimal among pure-state triples in arbitrary dimension, together
with further evidence.

For arbitrary $n$, we prove that every classical facet admitting a GPT
violation already admits a quantum violation in dimension at most three. A
four-state qubit example reveals a failure of submodularity, a feature absent
for $n=3$. Finally, as an outlook, we consider the prior-weighted
discrimination problem and exhibit a four-dimensional quantum triple whose
discrete profile $F$ is classical under uniform priors on all subfamilies, but
whose prior-weighted profile is nonclassical.
\end{abstract}

\section{Introduction}\label{sec:intro}

Minimum-error state discrimination usually asks a single question: given one
fixed ensemble, what is the best achievable success (or error) probability?
Here we ask it simultaneously for every subensemble. Given a fixed list of
states $s_1,\ldots,s_n$ of a classical, quantum, or an even more general
hypothetical system described by a general probabilistic theory (GPT) model,
for each nonempty $H\subseteq[n]$ let $F(H)$ denote the minimum unnormalized sum
of conditional error probabilities in guessing $i$ from a measurement, given
the promise that the system was prepared in state $s_i$ for some $i\in H$. We
call the resulting set function $F$ the \emph{discrimination profile} of the
list of states.

In \cite{Weiner2025}, where the present author studied possible patterns of
\emph{perfect} distinguishability, the more quantitative question of
characterizing \emph{all} properties of discrimination profiles was raised. Of
course, several constraints are immediate; as already noted there,
$0\le F(H)\le |H|-1$ for any nonempty $H$, and adjoining one more state
increases the sum of error probabilities by a quantity between zero and one:
$F(H)\le F(H\cup\{j\})\le F(H)+1$.

Less immediate, but equally model-independent, is the inequality
\begin{equation}\label{eq:intro-gluing}
 F(A\cup B)\ge F(A)+F(B)\qquad\text{if }|A\cap B|=1.
\end{equation}
It follows by reusing a measurement for $A\cup B$: outcomes outside each
subproblem are reassigned to its common hypothesis
(Proposition~\ref{prop:one-point} gives the details). Iterating this along a
tree gives
\begin{equation}\label{eq:intro-tree}
 F(H)\ge\max_{\mathcal T}\sum_{\{i,j\}\in E(\mathcal T)}F(ij),
\end{equation}
the maximum taken over spanning trees on $H$. This is the discrimination version of Hunter's classical union bound \cite{Hunter1976}, and does not seem to be widely used in the quantum state-discrimination literature. In particular, it strengthens the pairwise lower bound of Qiu \cite{Qiu2008}, which is also recalled, for example, in the introduction of Audenaert and Mosonyi \cite{AudenaertMosonyi2014}; see Section~\ref{sec:universal} for the precise comparison.

This raises a natural question: which restrictions on discrimination profiles
follow purely from the universal operational rules governing states,
measurements, and classical decision-making, and which require specifically
classical or quantum structure? ``Universal'' here does not mean weak or
obvious---as the example above shows, even elementary measurement manipulations
can yield useful bounds for quantum hypothesis testing.

Our guiding template is the local--quantum--no-signalling hierarchy of Bell
theory \cite{BrunnerEtAl2014}; here the analogous objects are the classical,
quantum, and GPT profile regions,
\begin{equation}\label{eq:hierarchy}
 \Creg_n\subseteq\Qreg_n\subseteq\Greg_n.
\end{equation}
The outer region $\Greg_n$ describes everything consistent with the general
operational framework; classical facets serve as witnesses of nonclassical
behavior; and quantum bounds measure how far quantum theory can push against
them. Ours is a single-system discrimination problem rather than a multipartite
correlation experiment, but the hierarchy plays the same organizing role.

Our main structural result is an exact description of the outer region
$\Greg_n$. A single fixed measurement yields a classical profile,
representable as a normalized coverage function; optimizing over all
measurements yields its upper envelope in success coordinates. The essential
converse is that \emph{every} finite upper envelope of this form is realizable:
a suitable Cartesian product of classical state spaces allows one to choose
among the corresponding classical experiments, and no affine measurement on
the product can outperform the best coordinate measurement for the subfamily
in question. Together these give both a realization theorem and a finite
linear-programming test, and show that $\Greg_n$, like $\Creg_n$, is a rational
polytope. The coverage-function representation itself belongs to established
theory \cite{StrokorbSchlather2015}; our contribution is to identify exactly how
it is reshaped by optimization over GPT measurements.

The three-state case already separates all three levels. We determine the
vertices and facets of both polytopes and find that a single additional
inequality cuts the classical region out of the GPT region:
\begin{equation}\label{eq:intro-beta}
 \beta:=F(123)-F(12)-F(13)-F(23)\le\beta_{\mathrm C}\equiv0.
\end{equation}
Its quantum supremum $\beta_{\mathrm Q}$ is a discrimination analogue of a
Tsirelson bound. We prove the strict separation
\begin{equation}\label{eq:intro-bounds}
 0=\beta_{\mathrm C}
 <\frac{3\sqrt3}{2}-2
 \le\beta_{\mathrm Q}\le\frac23
 <1=\beta_{\mathrm{GPT}}.
\end{equation}
The lower bound is attained by a qubit trine, and we show it is optimal both
among qubit triples and among arbitrary pure-state triples; whether it remains
optimal for all mixed-state triples is left open. The upper bound $2/3$, on the
other hand, is not optimal, and we present a modest improvement immediately
after its proof. Since the strict separation is already established, we do
not pursue further nonsharp refinements here; the remaining question is to
determine the optimal quantum bound.

For arbitrary $n$, the classical facets have an explicit M\"obius description.
We prove that every one of these facets that can be violated in a GPT can
already be violated by quantum states in dimension at most three. The proof
relies on two small examples together with repeated hypotheses. In particular,
a four-state qubit family shows that quantum success profiles need not be
submodular, although this property is universal for three hypotheses. This
facet theorem highlights a point of contrast with multipartite Bell scenarios,
where some local facets admit no quantum violation yet are violated by
no-signalling correlations \cite{AlmeidaEtAl2010}.

Finally, as an outlook, we consider a prior-weighted version of the
discrimination problem, where instead of assigning error values to nonempty
subsets of $[n]$, one tracks how the optimal error varies with the prior
probabilities themselves. In the homogeneous formulation used below, the
original discrete profile $F$ is recovered at indicator weights. We give an
example of a quantum triple whose discrete discrimination profile is entirely
classical, yet whose weighted profile is nonclassical---suggesting that the
weighted profile can reveal considerably more than its discrete shadow.

The remainder of the paper is organized as follows.
Sections~\ref{sec:setup}--\ref{sec:envelope} establish the universal constraints
and the classical and GPT characterizations.
Sections~\ref{sec:three}--\ref{sec:all-facets} treat the three-state regions,
quantum bounds, and higher-$n$ facets. Section~\ref{sec:weighted} concerns the
weighted version of the discrimination problem and the example above. The
appendix collects rigorous evidence for the trine conjecture, including its
optimality among pure-state triples.

\section{Profiles and universal measurement constraints}\label{sec:setup}

\subsection{Framework and normalization}

A finite-dimensional GPT state space is a nonempty compact convex set $K$.
An effect is an affine function $e:K\to[0,1]$, and a finite-outcome measurement
is a family $(e_x)_{x\in X}$ of nonnegative affine functions satisfying
$\sum_xe_x=1$. We use the no-restriction convention: every such measurement is
allowed. Classical state spaces are probability simplexes, while quantum state
spaces consist of density matrices on a finite-dimensional complex Hilbert
space, with measurements given by POVMs. Background on discrimination in GPTs
can be found in \cite{KimuraMiyaderaImai2009,NuidaKimuraMiyadera2010}.

Let $n\ge2$ and fix states $s_1,\ldots,s_n\in K$. The case $n=1$ has only the
zero error profile. For nonempty $H\subseteq[n]$, define
\begin{align}
 D(H)&:=\max_{(e_i)_{i\in H}}\sum_{i\in H}e_i(s_i),\label{eq:D-def}\\
 F(H)&:=|H|-D(H),\label{eq:F-def}
\end{align}
where the effects sum to the unit effect. Set $F(\varnothing)=D(\varnothing)=0$.
The maximum exists by compactness of the set of measurements with $|H|$
outcomes. Zero effects are allowed.

The quantities $D(H)$ and $F(H)$ are \emph{not probabilities}: the equal-prior
success and error probabilities are $D(H)/|H|$ and $F(H)/|H|$. The
unnormalized convention permits direct comparisons between different
subfamilies. In particular,
\begin{equation}\label{eq:normalization}
 F(\{i\})=0,\quad D(\{i\})=1,\qquad
 0\le F(H)\le |H|-1\quad(H\ne\varnothing).
\end{equation}
The upper bound follows by always guessing one fixed hypothesis. Restricting a
measurement to a smaller subproblem by relabeling the other outcomes, and
extending a measurement by a zero effect, give
\begin{equation}\label{eq:increment}
 0\le F(H\cup\{j\})-F(H)\le1\qquad(j\notin H).
\end{equation}

We call $F$ the \emph{discrimination profile}, or \emph{error profile}, and $D$
the \emph{success profile}. The states in the list need not be distinct.

\begin{definition}[Profile regions]\label{def:profile-regions}
Let $\mathcal I_n=\{H\subseteq[n]:|H|\ge2\}$. For a finite-dimensional compact
convex state space $K$, let $F_{K,\mathbf s}$ denote the error profile in
\eqref{eq:D-def}--\eqref{eq:F-def} of a list
$\mathbf s=(s_1,\ldots,s_n)\in K^n$, with all affine measurements allowed. Its
\emph{profile region} is
\begin{equation}\label{eq:fixed-space-region}
 \mathcal P_n(K):=
 \left\{\bigl(F_{K,\mathbf s}(H)\bigr)_{H\in\mathcal I_n}:
       \mathbf s\in K^n\right\}
 \subseteq\R^{\mathcal I_n}\cong\R^{2^n-n-1}.
\end{equation}
The \emph{GPT profile region} is
\begin{equation}\label{eq:Greg-definition}
 \Greg_n:=\bigcup_K\mathcal P_n(K),
\end{equation}
where $K$ ranges over all nonempty compact convex subsets of $\R^d$, for all
finite $d\ge1$.

Let
\[
 \Delta_m:=\left\{p\in\R^m:p_j\ge0,\ \sum_{j=1}^mp_j=1\right\},
 \qquad
 \mathcal D_d:=\{\rho\in M_d(\mathbb C):\rho\ge0,\ \Tr\rho=1\}.
\]
The \emph{classical} and \emph{quantum profile regions} are, respectively,
\begin{equation}\label{eq:Creg-Qreg-definition}
 \Creg_n:=\bigcup_{m\ge1}\mathcal P_n(\Delta_m),
 \qquad
 \Qreg_n:=\bigcup_{d\ge1}\mathcal P_n(\mathcal D_d).
\end{equation}
All alphabet sizes and Hilbert-space dimensions in these unions are finite.
No closure is taken in any of the definitions.
\end{definition}

The regions are thus expressed in error coordinates, with the empty-set and
singleton coordinates omitted because they are fixed. We identify a full
profile $F$ with its coordinate vector when writing $F\in\Greg_n$, and
similarly for the other regions. The affine transformation $D(H)=|H|-F(H)$
converts every statement into success coordinates. Classical states embed as
diagonal density matrices, giving the inclusions~\eqref{eq:hierarchy}.

\begin{proposition}[Convexity]\label{prop:convex}
The three profile regions are convex.
\end{proposition}
\begin{proof}
Mix two preparations with probabilities $\lambda$ and $1-\lambda$, independent
of the hypothesis, and retain a perfectly readable classical flag indicating
the preparation. The flag can be read first, after which the optimal measurement
for its block is performed. Conversely, conditioning any measurement on the
flag bounds its success by the corresponding weighted sum of block optima.
Thus the new profile is $\lambda F_1+(1-\lambda)F_2$.
For quantum states this is the block-diagonal construction
$\rho_i=\lambda\rho_i^{(1)}\oplus(1-\lambda)\rho_i^{(2)}$; the classical and
GPT constructions are the corresponding flagged direct sums.
\end{proof}

\subsection{One-point gluing and the tree bound}\label{sec:universal}

\begin{proposition}[One-point gluing]\label{prop:one-point}
If $A\cap B=\{j\}$, then
\begin{equation}\label{eq:one-point}
 F(A\cup B)\ge F(A)+F(B).
\end{equation}
\end{proposition}
\begin{proof}
Choose an optimal measurement $(e_i)_{i\in A\cup B}$ for the union. For the
$A$-problem, assign every outcome in $B\setminus\{j\}$ to $j$; for the
$B$-problem, assign every outcome in $A\setminus\{j\}$ to $j$.
If the success sums of these two measurements are $d_A$ and $d_B$, then
\[
 d_A+d_B
 =\sum_{i\in A\cup B}e_i(s_i)+\sum_{i\in A\cup B}e_i(s_j)
 =D(A\cup B)+1.
\]
Their error sums therefore add up to $F(A\cup B)$. Minimizing each subproblem
separately proves the claim.
\end{proof}

\begin{corollary}[Maximum-spanning-tree bound]\label{cor:tree}
For $|H|=m\ge2$,
\begin{equation}\label{eq:tree}
 F(H)\ge\max_{\mathcal T}\sum_{\{i,j\}\in E(\mathcal T)}F(ij)
 \ge\frac{2}{m}\sum_{\{i,j\}\subseteq H}F(ij).
\end{equation}
\end{corollary}
\begin{proof}
Remove a leaf of a spanning tree and apply Proposition~\ref{prop:one-point}
to its incident edge and the remaining vertex set. Induction proves the first
inequality. In a uniformly random labeled spanning tree, symmetry and the
number $m-1$ of edges imply that each of the $\binom m2$ edges occurs with
probability $2/m$. Taking the average proves the second inequality.
\end{proof}

For quantum states, the pairwise bound recalled in
\cite[Eq.~(6)]{AudenaertMosonyi2014}, due to Qiu \cite{Qiu2008}, is
\[
 F(H)\ge\frac{1}{m-1}\sum_{\{i,j\}\subseteq H}F(ij).
\]
Since $2/m\ge1/(m-1)$, \eqref{eq:tree} is at least as strong, and its averaged
form is strictly stronger when $m>2$ and the sum of pairwise errors is positive.
For $m$ identical states, the tree bound gives the exact value $m-1$, whereas
this pairwise-average bound gives $m/2$.
The underlying tree inequality is classical, not a new probability-union
inequality: the coverage representation below identifies it with Hunter's
bound \cite{Hunter1976}. Its application here emphasizes that useful quantum
bounds may follow without any specifically quantum inequality.

\section{A fixed measurement and classical profiles}\label{sec:coverage}

Given a measurement $M=(e_x)_{x\in X}$, let $p_i^M(x)=e_x(s_i)$. If the
measurement is fixed but its classical decoding may depend on $H$, its optimal
success is
\begin{equation}\label{eq:theta-M}
 \theta_M(H)=\sum_{x\in X}\max_{i\in H}p_i^M(x),
 \qquad\theta_M(\varnothing)=0.
\end{equation}
Indeed, after observing $x$, one chooses an index maximizing its likelihood.
Consequently,
\begin{equation}\label{eq:decode}
 D(H)=\max_M\theta_M(H).
\end{equation}
An $H$-indexed discrimination measurement is itself one of the measurements on
the right, and optimal decoding of an arbitrary measurement gives an allowed
$H$-indexed measurement.

Define the coefficient polytope
\begin{equation}\label{eq:Coeff}
 \Coeff_n:=\left\{(c_B)_{\varnothing\ne B\subseteq[n]}:
 c_B\ge0,\quad\sum_{B\ni i}c_B=1\quad(i\in[n])\right\}.
\end{equation}
Every coordinate lies in $[0,1]$. For $c\in\Coeff_n$, set
\begin{equation}\label{eq:coverage}
 \theta_c(H):=\sum_{B:B\cap H\ne\varnothing}c_B.
\end{equation}
These are the normalized coverage functions: $\theta_c(\varnothing)=0$ and
$\theta_c(\{i\})=1$. They also occur as finite extremal-coefficient functions;
see \cite{StrokorbSchlather2015} for their intrinsic characterization by
complete alternation.

\begin{proposition}[One measurement and classical realizability]
\label{prop:one-measurement}
The functions $\theta_M$ obtained from a fixed measurement on $n$ states are
exactly the functions $\theta_c$ with $c\in\Coeff_n$. Every such function is
the \emph{full optimal success profile} of $n$ classical probability
distributions on at most $2^n-1$ outcomes.
\end{proposition}
\begin{proof}
For a fixed measurement, consider $\Omega=X\times[0,1]$, with counting measure
on $X$ and Lebesgue measure on $[0,1]$. Under the graph of each distribution put
\[
 A_i:=\{(x,u):0\le u\le p_i^M(x)\}.
\]
Then $\mu(A_i)=1$ and
$\theta_M(H)=\mu(\bigcup_{i\in H}A_i)$.
To describe the overlaps, assign to each point the set of graphs under which
it lies:
\[
 \omega\longmapsto B(\omega):=\{i\in[n]:\omega\in A_i\}.
\]
For a nonempty subset $B$, let $c_B$ be the measure of the region
$\{\omega:B(\omega)=B\}$. These disjoint regions partition
$\bigcup_i A_i$. A region contributes to the union indexed by $H$ exactly when
$B\cap H\ne\varnothing$, giving \eqref{eq:coverage}; summing the regions
containing $i$ gives the marginal equalities in \eqref{eq:Coeff}.

Conversely, for $c\in\Coeff_n$, use the nonempty subsets $B\subseteq[n]$ as
outcomes and define classical probability distributions
\begin{equation}\label{eq:classical-distributions}
 p_i(B):=c_B\1_{\{i\in B\}}.
\end{equation}
Their normalization follows from \eqref{eq:Coeff}. Direct observation of $B$,
followed by optimal decoding, gives
\[
 \sum_B\max_{i\in H}p_i(B)=\theta_c(H).
\]
No stochastic processing of a classical outcome can improve this optimum, so
this is the full classical success profile, not merely the performance of a
particular measurement.
\end{proof}

Write
\begin{equation}\label{eq:E-c}
 E_c(H):=|H|-\theta_c(H)
       =\sum_{\varnothing\ne B\subseteq[n]}c_B\bigl(|B\cap H|-1\bigr)_+,
\end{equation}
where $u_+=\max\{u,0\}$. The last equality follows by summing the marginal
constraints over $i\in H$.

\begin{corollary}[The classical region]\label{cor:classical-region}
In error coordinates,
\begin{equation}\label{eq:Creg}
 \Creg_n=\{(E_c(H))_{|H|\ge2}:c\in\Coeff_n\}.
\end{equation}
Equivalently, in success coordinates the classical region is precisely
$\{\theta_c:c\in\Coeff_n\}$. In particular, $\Creg_n$ is a rational polytope.
\end{corollary}

The representation also applies to classical distributions on an arbitrary
measurable space: take densities with respect to their sum and use the same
subgraph construction. Thus an infinite classical alphabet does not enlarge
the discrete profile region.

\subsection{Signed coverage coefficients and classical facets}

The coverage coefficients are uniquely determined by the profile, even when
nonnegativity fails. This gives a useful coordinate system beyond the
classical region.

\begin{proposition}[Classical facets]\label{prop:all-classical-facets}
Let $D(\varnothing)=0$ and $D(\{i\})=1$. For each nonempty $B\subseteq[n]$, put
\begin{equation}\label{eq:signed-coefficients}
 \widehat c_B(D):=-\sum_{A\subseteq B}(-1)^{|B|-|A|}D([n]\setminus A).
\end{equation}
These are the unique signed coefficients satisfying
\begin{equation}\label{eq:signed-coverage}
 D(H)=\sum_{B:B\cap H\ne\varnothing}\widehat c_B(D),\qquad
 \sum_{B\ni i}\widehat c_B(D)=1.
\end{equation}
For $n\ge3$, the classical region has dimension $2^n-n-1$, and its distinct
facets are exactly
\begin{equation}\label{eq:classical-facets-all}
 \widehat c_B(D)\ge0\qquad(\varnothing\ne B\subseteq[n]).
\end{equation}
\end{proposition}
\begin{proof}
Apply M\"obius inversion to $U(A)=D([n])-D([n]\setminus A)$, with
$U(\varnothing)=0$. This gives
$U(A)=\sum_{\varnothing\ne B\subseteq A}\widehat c_B(D)$, and hence
\eqref{eq:signed-coverage}; the marginal identities follow by taking
$H=\{i\}$. Thus the coverage map is an affine isomorphism between the
coefficient space with its $n$ marginal constraints and the normalized profile
space. The marginal constraints are independent, and $c_B=2^{1-n}$ is a
strictly positive feasible point. The dimension follows.

It remains to check that every coordinate inequality defines a facet.
For $|B|\ge2$, set $c_B=0$, give every other nonsingleton coefficient a common
small value $\varepsilon>0$, and complete the marginal constraints with
positive singleton coefficients. For $B=\{i\}$, put $c_i=0$ and give every
nonsingleton coefficient except the pairs $\{i,j\}$ the value $\varepsilon$.
Set
\[
 c_{ij}=\frac{1-(2^{n-1}-n)\varepsilon}{n-1}\qquad(j\ne i),
\]
and complete the remaining marginals with $c_j$ for $j\ne i$.
As $\varepsilon\downarrow0$, these singleton coefficients tend to
$(n-2)/(n-1)>0$. Thus in either case there is a feasible point with precisely
the specified coefficient zero and all the others positive. Each coordinate
hyperplane therefore supports a distinct facet, and the defining
inequalities of $\Coeff_n$ leave no other facets.
\end{proof}

For example, the singleton inequalities are always universal:
\begin{equation}\label{eq:singleton-universal}
 \widehat c_{\{i\}}(D)=D([n])-D([n]\setminus\{i\})\ge0.
\end{equation}
This is just monotonicity of optimal success. Section~\ref{sec:all-facets}
determines which of the other classical facets can be violated by quantum
states.

\section{Exact characterization of GPT profiles}\label{sec:envelope}

\begin{theorem}[Envelope characterization]\label{thm:envelope}
Let $n\ge2$, let $F:2^{[n]}\to\R$ satisfy $F(\varnothing)=0$, and set
$D(H)=|H|-F(H)$. The following are equivalent:
\begin{enumerate}[label=(\roman*)]
\item $F$ is the discrimination-error profile of $n$ states in a
finite-dimensional GPT.
\item There exist $c^{(1)},\ldots,c^{(r)}\in\Coeff_n$ such that
\begin{equation}\label{eq:envelope-D}
 D(H)=\max_{1\le\alpha\le r}\theta_{c^{(\alpha)}}(H)
 \quad\text{for every }H\subseteq[n].
\end{equation}
\item There exist $c^{(1)},\ldots,c^{(r)}\in\Coeff_n$ such that
\begin{equation}\label{eq:envelope-F}
 F(H)=\min_{1\le\alpha\le r}E_{c^{(\alpha)}}(H)
 \quad\text{for every }H\subseteq[n].
\end{equation}
\end{enumerate}
One may always take $1\le r\le2^n-n-1$. Every such profile has a realization
on a Cartesian product of classical simplexes, with all affine measurements
allowed.
\end{theorem}

\subsection{Necessity}

For each $H$ with $|H|\ge2$, choose a measurement optimal for $H$ and let
$c^H\in\Coeff_n$ be its coefficient vector from
Proposition~\ref{prop:one-measurement}. Then
\[
 \theta_{c^H}(K)\le D(K)\quad\text{for all }K,
 \qquad\theta_{c^H}(H)=D(H).
\]
Taking the maximum over these $2^n-n-1$ supporting functions proves
\eqref{eq:envelope-D}. No additional functions are needed for singletons or the
empty set, where all normalized coverage functions agree. Subtracting from
$|H|$ proves the equivalence with \eqref{eq:envelope-F}.

This argument does not require that all affine effects be available in the
original model. It applies to restricted measurement sets closed under
classical decoding; if optima are only suprema, choose an approximating sequence
for each $H$ and take a convergent subsequence in the compact polytope
$\Coeff_n$. The limiting function remains dominated by $D$ everywhere and
attains $D(H)$. Such restrictions therefore do not enlarge the outer profile
region obtained here.

\subsection{Measurements on a Cartesian product}

\begin{lemma}[Product-measurement decomposition]\label{lem:product-measurement}
	Let $K=\prod_{\alpha=1}^r K_\alpha$ be a finite Cartesian product of nonempty
	compact convex sets. Every finite-outcome affine measurement $(e_y)_{y\in Y}$
	on $K$ is a convex mixture of measurements depending on a single coordinate:
	\begin{equation}\label{eq:product-measurement}
		e_y(x_1,\ldots,x_r)=\sum_{\alpha=1}^r
		\lambda_\alpha e_y^{(\alpha)}(x_\alpha),
		\qquad \lambda_\alpha\ge0,\quad\sum_\alpha\lambda_\alpha=1,
	\end{equation}
	where $(e_y^{(\alpha)})_{y\in Y}$ is a measurement on $K_\alpha$. 
\end{lemma}

\begin{proof}
	Fix a basepoint in each $K_\alpha$; this identifies the affine hull of $K$
	with a direct sum of vector spaces, so the linear part of any affine
	functional on $K$ splits as a sum of linear functionals on the summands.
	Hence each $e_y$ can be written as
	\[
	e_y(x)=\sum_{\alpha=1}^r h_y^{(\alpha)}(x_\alpha),
	\]
	with $h_y^{(\alpha)}$ affine on $K_\alpha$. This decomposes the linear part
	of $e_y$ uniquely; the constant terms of the $h_y^{(\alpha)}$ are free
	subject to summing to the (fixed) constant term of $e_y$, giving $r-1$
	free parameters. We use these to impose the $r-1$ equations
	$\min_{K_1}h_y^{(1)}=\cdots=\min_{K_r}h_y^{(r)}=:m_y$, which is therefore
	always achievable.
	
	Because $K$ is a full Cartesian product, minimization over $K$ splits coordinatewise:
	\[
	\min_{x\in K}e_y(x)=\sum_{\alpha=1}^r\min_{K_\alpha}h_y^{(\alpha)}=rm_y.
	\]
	This implies that $m_y\ge0$ as $e_y^{(\alpha)}\geq 0$.
	Thus each coordinate function is
	nonnegative.
	
	Since $\sum_ye_y=1$, varying only coordinate $\alpha$ shows that
	$\sum_yh_y^{(\alpha)}$ is a constant $\lambda_\alpha$. These constants are
	nonnegative and sum to one. For $\lambda_\alpha>0$, define
	$e_y^{(\alpha)}=h_y^{(\alpha)}/\lambda_\alpha$; zero-weight terms can be filled with any measurement. This proves \eqref{eq:product-measurement}.
	
\end{proof}

\subsection{Converse realization and dimension}

For each $c^{(\alpha)}$, take the classical distributions
$p_i^{(\alpha)}$ from \eqref{eq:classical-distributions}, on an alphabet
$X_\alpha$ with zero-probability outcomes omitted. Define
\begin{equation}\label{eq:product-realization}
 K=\prod_{\alpha=1}^r\Delta(X_\alpha),\qquad
 s_i=(p_i^{(1)},\ldots,p_i^{(r)}).
\end{equation}
Reading coordinate $\alpha$ gives success $\theta_{c^{(\alpha)}}(H)$, so the
optimal success is at least their maximum. Conversely,
Lemma~\ref{lem:product-measurement} and convexity of a pointwise maximum give,
for any measurement on $K$,
\begin{align*}
 \sum_y\max_{i\in H}e_y(s_i)
 &\le\sum_\alpha\lambda_\alpha
       \sum_y\max_{i\in H}e_y^{(\alpha)}(p_i^{(\alpha)})\\
 &\le\sum_\alpha\lambda_\alpha\theta_{c^{(\alpha)}}(H)
 \le\max_\alpha\theta_{c^{(\alpha)}}(H).
\end{align*}
The middle inequality uses the optimality of direct observation in a classical
simplex. This proves the converse in Theorem~\ref{thm:envelope}. Notice that
$K$ is a \emph{Cartesian product of state spaces}, not a classical joint
probability simplex: its affine effects do not allow simultaneous readout of
all coordinates.

The construction gives the explicit, generally nonoptimal dimension bound
\begin{equation}\label{eq:dimension-bound}
 \dim K=\sum_\alpha(|X_\alpha|-1)
 \le(2^n-n-1)(2^n-2).
\end{equation}
Its qualitative specialization explains the distinction between unrestricted
realizability and economical realizability. Let $\mathcal A$ be a downward-
closed family of subsets of $[n]$ containing all singletons, and let
$\mathcal M$ be its maximal members. For each $M\in\mathcal M$, use a simplex
with vertices $(v_i)_{i\in M}$; assign $v_i$ to hypothesis $i\in M$ and its
barycenter to every $i\notin M$. A subfamily with at least two members is
perfectly distinguishable in this coordinate exactly when it is contained in
$M$. The product lemma then shows that the zero set of its full error profile
is precisely $\mathcal A$.

This recovers the existence conclusion of \cite{Weiner2025}, but its dimension
is $\sum_{M\in\mathcal M}(|M|-1)$, potentially exponentially larger than the
$n-1$ dimensions of the construction in that work. Dimension-efficient
realization remains a substantive geometric problem. Cartesian products and
related compression questions are also central in the study of higher rank
antipodality by Nasz\'odi, Szil\'agyi, and Weiner
\cite{NaszodiSzilagyiWeiner2025}. For quantitative profiles, a natural refinement
is to minimize the dimension needed to reproduce all values of $F$, not only
its zero set.

\subsection{A finite linear program and the outer polytope}

\begin{theorem}[Supporting linear program]\label{thm:LP}
Let $F(\varnothing)=F(\{i\})=0$. Then $F\in\Greg_n$ if and only if for each
$H$ with $|H|\ge2$ there exists $c^H\in\Coeff_n$ such that
\begin{equation}\label{eq:LP}
 E_{c^H}(K)\ge F(K)\quad(|K|\ge2),\qquad E_{c^H}(H)=F(H).
\end{equation}
In particular, $\Greg_n$ is a bounded rational polytope.
\end{theorem}
\begin{proof}
For an envelope, choose a minimizing coefficient vector at $H$. Conversely,
\eqref{eq:LP} gives $F(K)=\min_{|H|\ge2}E_{c^H}(K)$, including the fixed
singleton and empty-set coordinates. Theorem~\ref{thm:envelope} applies.
All the coefficient constraints and \eqref{eq:LP} are linear with integer
coefficients. In their joint feasible set every $c_B^H$ lies in $[0,1]$, and
each $F(H)$ equals an $E_{c^H}(H)\in[0,|H|-1]$. Thus this is a bounded rational
polyhedron. Its projection onto the $F$-coordinates is a rational polytope.
\end{proof}

The theorem characterizes all universal \emph{linear} discrimination
inequalities as the valid inequalities of an explicit finite polytope. It also
supplies a realization certificate, rather than merely an outer relaxation.
The number of variables grows exponentially; no claim of polynomial complexity
in $n$ is intended.

\begin{proposition}[General overlap inequality]\label{prop:overlap}
For arbitrary $A,B\subseteq[n]$,
\begin{equation}\label{eq:overlap}
 F(A\cup B)\ge F(A)+F(B)-(|A\cap B|-1)_+.
\end{equation}
\end{proposition}
\begin{proof}
Each coverage function is submodular, so its error function $E_c$ is
supermodular. Choose a supporting $E_c$ with $E_c(A\cup B)=F(A\cup B)$.
Then
\[
 F(A\cup B)\ge E_c(A)+E_c(B)-E_c(A\cap B)
 \ge F(A)+F(B)-(|A\cap B|-1)_+.
\]
Here $E_c(I)\le(|I|-1)_+$ for all $I$. A support can also be chosen when the
union has fewer than two elements, since all clauses agree there.
\end{proof}

Unlike its individual supporting functions, a GPT error profile need not itself
be supermodular. The controlled loss in \eqref{eq:overlap} is therefore
important; simply taking minima of classical supermodularity inequalities
would not prove supermodularity of $F$.

\section{The complete three-state geometry}\label{sec:three}

Put
\begin{equation}\label{eq:three-coordinates}
 x=F(12),\quad y=F(13),\quad z=F(23),\quad t=F(123).
\end{equation}
The remaining coordinates are fixed by normalization.

\begin{theorem}[The GPT polytope for three states]\label{thm:G3}
The region $\Greg_3\subset\R^4$ has the following irredundant facet description:
\begin{align}
 x,y,z&\ge0,\label{eq:G3-positive}\\
 t&\ge x+y,\quad t\ge x+z,\quad t\ge y+z,\label{eq:G3-lower}\\
 t&\le1+x,\quad t\le1+y,\quad t\le1+z.\label{eq:G3-upper}
\end{align}
Its vertices are
\begin{align*}
 O&=(0,0,0,0),& X&=(1,0,0,1),&Y&=(0,1,0,1),\\
 Z&=(0,0,1,1),& W&=(1,1,1,2),&P&=(0,0,0,1).
\end{align*}
\end{theorem}
\begin{proof}
Nonnegativity, one-point gluing, and the increment bounds prove necessity.
To prove sufficiency, let $(x,y,z,t)$ satisfy the displayed inequalities.
Choose
\[
 \max\{0,x+y+z-t,t-1\}\le w\le\min\{x,y,z\}.
\]
This interval is nonempty: the lower-sum inequalities bound $x+y+z-t$ by each
of $x,y,z$, and the upper inequalities do the same for $t-1$. The coefficients
\begin{align*}
 \lambda_W&=w,&\lambda_X&=x-w,&\lambda_Y&=y-w,&\lambda_Z&=z-w,\\
 \lambda_P&=t-x-y-z+w,&\lambda_O&=1-t+w
\end{align*}
are nonnegative, sum to one, and give the stated point as a convex combination
of $O,X,Y,Z,W,P$.

The first five points have classical realizations: $O$ represents three
distinct point masses; $X$ represents equal first and second states supported
apart from the third; $Y,Z$ are its permutations; and $W$ represents three
identical states. For $P$, use the square $[0,1]^2$ and states
$(0,0),(1,0),(0,1)$. Every pair differs in one coordinate and is perfectly
distinguishable. Each binary coordinate has triple success at most two, and
Lemma~\ref{lem:product-measurement} shows that no affine measurement on the
square can do better. Coordinate readout attains two, hence its profile is $P$.
Convexity now proves sufficiency.

Each listed point is a vertex, as can be checked from four independent active
constraints. The polytope is four-dimensional. For facetness, representative
facets $x=0$, $t=x+y$, and $t=1+x$ contain, respectively,
\[
 \{O,Y,Z,P\},\qquad\{O,X,Y,W\},\qquad\{Y,Z,P,W\},
\]
four affinely independent points each. Permuting $x,y,z$ proves that all nine
inequalities define facets.
\end{proof}

\begin{theorem}[The classical polytope for three states]\label{thm:C3}
The classical region is
\begin{equation}\label{eq:C3-cut}
 \Creg_3=\Greg_3\cap\{t\le x+y+z\}.
\end{equation}
Its seven facets are the six inequalities
\eqref{eq:G3-lower}--\eqref{eq:G3-upper} and $t\le x+y+z$.
Its vertices are $O,X,Y,Z,W$ and
\begin{equation}\label{eq:M-vertex}
 M=\left(\frac12,\frac12,\frac12,\frac32\right).
\end{equation}
\end{theorem}
\begin{proof}
For a classical profile, the seven coverage coefficients are necessarily
\begin{align}
 c_{123}&=x+y+z-t,\nonumber\\
 c_{12}&=t-y-z,&c_{13}&=t-x-z,&c_{23}&=t-x-y,\label{eq:inverse-coeff}\\
 c_1&=1+z-t,&c_2&=1+y-t,&c_3&=1+x-t.\nonumber
\end{align}
Conversely, these formulas satisfy all three marginal equalities, and their
nonnegativity is exactly the seven claimed inequalities. The classical
characterization in Corollary~\ref{cor:classical-region} proves the result.
These inequalities already imply $x,y,z\ge0$, since, for example,
$t\ge y+z$ and $t\le x+y+z$ imply $x\ge0$.

For the vertex description, the cutting hyperplane passes through $O,X,Y,Z$;
$P$ lies strictly above it and $W$ strictly below it. Therefore the only possible
new vertex comes from cutting the edge $PW$. This is indeed an edge, being the
intersection of the three upper facets; it is cut at its midpoint $M$.
All five retained GPT vertices remain extreme. Equivalently, $M$ has the
classical realization $c_{12}=c_{13}=c_{23}=1/2$ with all other coefficients
zero. The new facet contains $O,X,Y,Z$. The lower facets contain the independent
sets used above; for the upper facets, replacing $P$ by $M$ in those sets
preserves affine independence. Hence all seven inequalities are facets.
\end{proof}

\begin{remark}
	The facet description in Theorem~\ref{thm:C3} is also an immediate specialization of Proposition~\ref{prop:all-classical-facets}. Indeed, for $n=3$ the seven inequalities $\widehat c_B(D)\ge0$, $\varnothing\ne B\subseteq[3]$, are precisely the nonnegativity conditions displayed in \eqref{eq:inverse-coeff}. Six of them are already valid throughout $\Greg_3$ by Theorem~\ref{thm:G3}; the only additional classical constraint is
	
	$$
	\widehat c_{123}=x+y+z-t\ge0,
	$$
	or equivalently $\beta\le0$. We have retained the direct calculation above because it makes the three-state geometry and the subsequent vertex description transparent.
\end{remark}

\begin{corollary}[A complete three-state profile witness]\label{cor:unique-facet}
For a GPT profile on three hypotheses, define
\begin{equation}\label{eq:beta}
 \beta=t-x-y-z.
\end{equation}
The profile is classical if and only if $\beta\le0$. The maxima over the
classical and GPT regions are, respectively, $0$ and $1$; the latter is
attained only at $P$.
\end{corollary}

In \eqref{eq:inverse-coeff}, every coefficient except $c_{123}=-\beta$ is
nonnegative throughout $\Greg_3$. The nonclassical facet therefore detects
exactly a negative three-way overlap in the formal classical representation.
This is a statement about the \emph{profile}: a noncommuting quantum family may
still have a classical profile because another family of classical states
reproduces all of these finitely many numbers.

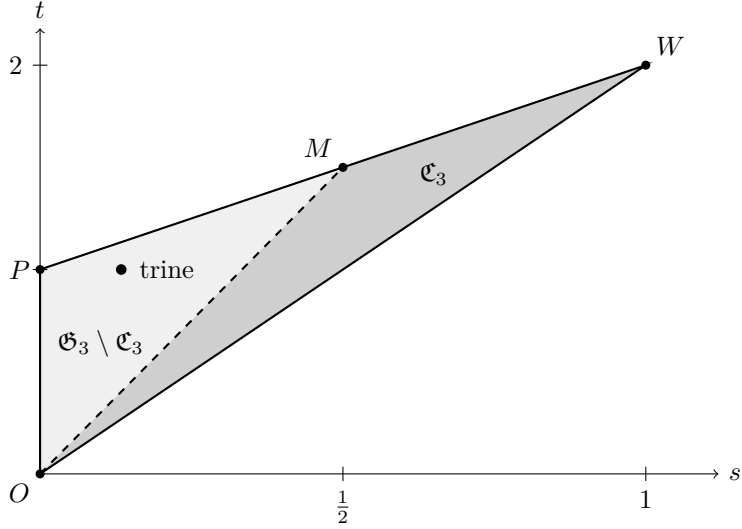
\begin{figure}[t]
\centering
\begin{tikzpicture}[x=8cm,y=2.7cm,font=\small]
 \fill[gray!12] (0,0)--(0,1)--(1,2)--cycle;
 \fill[gray!38] (0,0)--(.5,1.5)--(1,2)--cycle;
 \draw[thick] (0,0)--(0,1)--(1,2)--cycle;
 \draw[dashed,thick] (0,0)--(.5,1.5);
 \draw[->] (0,0)--(1.12,0) node[right] {$s$};
 \draw[->] (0,0)--(0,2.18) node[above] {$t$};
 \draw (.5,.035)--(.5,-.035) node[below] {$\tfrac12$};
 \draw (1,.035)--(1,-.035) node[below] {$1$};
 \draw (.012,1)--(-.012,1);
 \draw (.012,2)--(-.012,2) node[left] {$2$};
 \fill (0,0) circle[radius=1.7pt] node[below left] {$O$};
 \fill (0,1) circle[radius=1.7pt] node[left] {$P$};
 \fill (.5,1.5) circle[radius=1.7pt] node[above left] {$M$};
 \fill (1,2) circle[radius=1.7pt] node[above right] {$W$};
 \node at (.65,1.47) {$\Creg_3$};
 \node at (.10,.64) {$\Greg_3\setminus\Creg_3$};
 \fill (0.1339745962,1) circle[radius=2pt] node[right=3pt] {trine};
\end{tikzpicture}
\caption{The symmetric slice $x=y=z=s$. The GPT region is the triangle $OPW$;
its classical part is the darker triangle $OMW$. The dashed segment $OM$ is
$t=3s$, or $\beta=0$. The trine has $s=1-\sqrt3/2$ and $t=1$.
No boundary of the quantum region is asserted in this picture.}
\label{fig:slice}
\end{figure}

\section{Quantum violation and dimension-independent bounds}\label{sec:quantum}

\subsection{Operator formulation}

For density matrices $\rho_i$, denote the unnormalized optimal success by
$\Sopt_H=D(H)$. Its primal and dual semidefinite programs are
\begin{equation}\label{eq:quantum-SDP}
 \Sopt_H
 =\max_{\substack{M_i\ge0\ (i\in H)\\\sum_{i\in H}M_i=\id}}
       \sum_{i\in H}\Tr(\rho_iM_i)
 =\min_{\substack{Y=Y^*\\Y\ge\rho_i\ (i\in H)}}\Tr Y.
\end{equation}
Both are strictly feasible, so strong duality holds. For background see
\cite{Helstrom1976,BaeKwek2015}. For a pair, the Helstrom formula gives
\begin{equation}\label{eq:Helstrom}
 F(ij)=1-\delta_{ij},\qquad
 \delta_{ij}:=\frac12\|\rho_i-\rho_j\|_1.
\end{equation}
For a triple, abbreviate $\Sopt=\Sopt_{123}$ and put
$T=\delta_{12}+\delta_{13}+\delta_{23}$. Then
\begin{equation}\label{eq:beta-operator}
 \beta=T-\Sopt,\qquad 0\le T\le3.
\end{equation}
We write $\beta_{\mathrm Q}$ for the supremum of this quantity over all triples
in arbitrary finite dimensions.

\subsection{Qubits and the trine}

\begin{proposition}[The exact qubit value]\label{prop:qubit}
For arbitrary qubit states,
\begin{equation}\label{eq:qubit-bound}
 \beta\le\btr=\frac{3\sqrt3}{2}-2.
\end{equation}
Equality holds precisely for a pure trine, up to unitary conjugation and
permutation of the states.
\end{proposition}
\begin{proof}
Use the Bloch representation
\[
 \rho_i=\frac12(\id+\vec r_i\cdot\vec\sigma),\qquad
 \|\vec r_i\|\le1,
\]
where $\vec\sigma=(\sigma_x,\sigma_y,\sigma_z)$ is the vector of Pauli matrices
and $\|\cdot\|$ is the Euclidean norm. If $L$ is the perimeter of the triangle
with vertices $\vec r_1,\vec r_2,\vec r_3$, then
$\|\rho_i-\rho_j\|_1=\|\vec r_i-\vec r_j\|$, so $T=L/2$.

To evaluate the dual in \eqref{eq:quantum-SDP}, write
$Y=(a\id+\vec b\cdot\vec\sigma)/2$. Its constraints are exactly
$a-1\ge\|\vec b-\vec r_i\|$, while $\Tr Y=a$. Thus, if $R$ is the radius of
the smallest enclosing ball of the three Bloch vectors,
\begin{equation}\label{eq:qubit-geometry}
 \Sopt=1+R,\qquad\beta=\frac L2-R-1.
\end{equation}

We claim $L\le3\sqrt3R$, with equality for a nondegenerate triangle only when
it is equilateral. For an acute triangle, the smallest enclosing circle is its
circumcircle. If its interior angles are $A,B,C$, then
\[
 L=2R(\sin A+\sin B+\sin C)\le3\sqrt3R
\]
by concavity of sine and $A+B+C=\pi$. For a right or obtuse triangle, its longest
side has length $a=2R$. The remaining side lengths satisfy $b^2+c^2\le a^2$,
so
\[
 L=a+b+c\le(1+\sqrt2)a=2(1+\sqrt2)R<3\sqrt3R.
\]
For collinear points $L=4R$; coincident points cause no difficulty. This
separation of cases avoids identifying the circumradius with the enclosing
radius when the triangle is obtuse.

Since the Bloch ball itself contains the three points, $R\le1$. Therefore
\[
 \beta\le\left(\frac{3\sqrt3}{2}-1\right)R-1
 \le\frac{3\sqrt3}{2}-2.
\]
Equality requires an equilateral triangle of circumradius one. Write its
vertices as $\vec c+\vec u_i$, where $\sum_i\vec u_i=0$ and
$\|\vec u_i\|=1$. Their average squared norm is $\|\vec c\|^2+1\le1$,
forcing $\vec c=0$. Hence the three states are pure and form a trine.
Conversely, for a pure trine $R=1$ and $L=3\sqrt3$, so both inequalities
are equalities.
\end{proof}

For later use, fix trine states $\tau_i$ with Bloch vectors
\begin{equation}\label{eq:trine-vectors}
 \vec r_1=(1,0,0),\quad
 \vec r_2=(-1/2,\sqrt3/2,0),\quad
 \vec r_3=(-1/2,-\sqrt3/2,0).
\end{equation}
Their pairwise trace distances are $\delta_{ij}=\sqrt3/2$, and
\eqref{eq:qubit-geometry} gives $\Sopt=1+R=2$. Thus
\begin{equation}\label{eq:trine-profile}
 F(ij)=1-\frac{\sqrt3}{2},\qquad F(123)=1,\qquad
 \beta=\btr>0.
\end{equation}
The trine profile is therefore nonclassical.

\subsection{An elementary strict quantum--GPT separation}

A higher-dimensional quantum triple cannot approach the GPT maximum $1$.
The following proof uses only a column-operator construction and
Hilbert--Schmidt and Schatten inequalities.

\begin{theorem}[A dimension-independent quadratic bound]\label{thm:two-thirds}
For every triple of density matrices,
\begin{equation}\label{eq:quadratic-bound}
 \beta\le q(T):=T-1-\frac{4T^2}{27}\le\frac23.
\end{equation}
\end{theorem}
\begin{proof}
Define
\[
 A=(\rho_1^2+\rho_2^2+\rho_3^2)^{1/2},\qquad\tau=\Tr A.
\]
Place the three density matrices above one another to obtain the column
operator $\mathcal R:\mathcal H\to\mathcal H^{\oplus3}$:
\begin{equation}\label{eq:column}
 \mathcal R=\begin{pmatrix}\rho_1\\\rho_2\\\rho_3\end{pmatrix},
 \qquad V=\begin{pmatrix}V_1\\V_2\\V_3\end{pmatrix},
 \qquad\mathcal R=VA.
\end{equation}
Here $\mathcal R=VA$ is the polar decomposition: its positive factor is exactly
the $A$ already defined, because $\mathcal R^*\mathcal R=A^2$. If $\Pi$ is the
support projection of $A$, then
\begin{equation}\label{eq:polar-identities}
 \rho_i=V_iA,\qquad\sum_iV_i^*V_i=\Pi.
\end{equation}
All $\rho_i$ are supported on $\Pi$.

We first prove the standard quadratic success estimate
\begin{equation}\label{eq:HC}
 \Sopt\ge\frac{\tau^2}{3},
\end{equation}
which is a form of the generalized Holevo--Curlander bound
\cite{Tyson2009}. The same operators give a short proof. The effects
$N_i=V_i^*V_i$ form a measurement on the common support; distribute
$\id-\Pi$ arbitrarily among them to obtain a POVM on $\mathcal H$.
Since $\sum_i\Tr(V_i^*\rho_i)=\Tr(\Pi A)=\tau$, Hilbert--Schmidt
Cauchy--Schwarz yields
\[
 \tau\le\sum_i|\Tr(V_i^*\rho_i)|
 \le\sum_i\sqrt{\Tr(\rho_iN_i)}
 \le\sqrt{3\sum_i\Tr(\rho_iN_i)}\le\sqrt{3\Sopt}.
\]
The middle estimate follows by applying Cauchy--Schwarz to
$V_i\rho_i^{1/2}$ and $\rho_i^{1/2}$, using $\Tr\rho_i=1$.

For $i<j$, put
$q_{ij}=\Tr\bigl(A(V_i-V_j)^*(V_i-V_j)\bigr)$. Schatten H\"older gives
\begin{align}
 \|\rho_i-\rho_j\|_1
 &=\|(V_i-V_j)A\|_1\nonumber\\
 &=\|\bigl((V_i-V_j)A^{1/2}\bigr)A^{1/2}\|_1\nonumber\\
 &\le\|(V_i-V_j)A^{1/2}\|_2\|A^{1/2}\|_2
 =\sqrt{\tau q_{ij}}.\label{eq:holder}
\end{align}
Let $W=V_1+V_2+V_3$. Then
\[
 \sum_{i<j}(V_i-V_j)^*(V_i-V_j)=3\Pi-W^*W,
 \qquad\Tr(WA)=\sum_i\Tr\rho_i=3.
\]
Apply Cauchy--Schwarz to the Hilbert--Schmidt product
$\langle A^{1/2},WA^{1/2}\rangle=\Tr(WA)$ to obtain
\[
 9\le\tau\Tr(AW^*W),\qquad
 \sum_{i<j}q_{ij}\le3\tau-\frac9\tau.
\]
Consequently,
\[
 T\le\frac12\sqrt{3\tau\sum_{i<j}q_{ij}}
 \le\frac32\sqrt{\tau^2-3},
\]
so $\tau^2\ge3+4T^2/9$. Combining this with \eqref{eq:HC} proves the first
inequality in \eqref{eq:quadratic-bound}. Finally,
$q'(T)=1-8T/27>0$ for $0\le T\le3$, and $q(3)=2/3$.
\end{proof}

Together with the trine, this proves
$\Creg_3\subsetneq\Qreg_3\subsetneq\Greg_3$ by a single linear functional,
with a dimension-independent gap from the outer maximum.

\subsection{Why the endpoint bound is not sharp}

The endpoint $T=3$ in the last proof means that all pairs are perfectly
distinguishable. Their supports are then pairwise orthogonal, so
$\Sopt=3$ and $\beta=0$, rather than $2/3$. This observation can be made
quantitative using an upper bound on the joint error.

Let $f_{ij}=\|\sqrt{\rho_i}\sqrt{\rho_j}\|_1$ be the root fidelity. The bound of
Audenaert and Mosonyi \cite[Theorem~3.1]{AudenaertMosonyi2014} gives
$F(123)\le\sum_{i<j}f_{ij}$. The upper fidelity--trace-distance inequality and
concavity imply
\[
 \sum_{i<j}f_{ij}
 \le\sum_{i<j}\sqrt{1-\delta_{ij}^2}
 \le\sqrt{9-T^2}.
\]
Thus
\begin{equation}\label{eq:AM-beta}
 \beta\le h(T):=T-3+\sqrt{9-T^2}.
\end{equation}
Together with \eqref{eq:quadratic-bound}, this yields the explicit improvement
\begin{equation}\label{eq:AM-refinement}
 \beta\le\frac{3\sqrt{15}-9}{4}\simeq0.654738<\frac23.
\end{equation}
Indeed, $q$ is increasing, $h$ is decreasing for $T\ge3/\sqrt2$, and they meet
at $T_*=3\sqrt{15}/4$. Use $q(T)\le q(T_*)$ below $T_*$ and
$h(T)\le h(T_*)$ above it. This proves explicitly that the $2/3$ bound is
nonoptimal without resolving the sharp quantum value.

\subsection{The remaining conjecture}

The preceding bounds leave the interval
\begin{equation}\label{eq:remaining-interval}
 0.598076\ldots=\btr
 \le\beta_{\mathrm Q}
 \le\frac{3\sqrt{15}-9}{4}=0.654737\ldots.
\end{equation}

\begin{conjecture}[The trine bound]\label{conj:trine}
For every triple of quantum states in arbitrary finite dimension,
$\beta\le3\sqrt3/2-2$.
\end{conjecture}

Beyond the qubit proof, Appendix~\ref{app:evidence} establishes this bound for
all pure triples, with equality only for an embedded trine. It also proves
that removing a common positive semidefinite component cannot decrease
$\beta$, and describes mixed triples attaining the trine value. The latter
point is important: the conjecture concerns the optimal value, not the claim
that every maximizing triple must be pure.

\section{Quantum violations of classical facets beyond three states}
\label{sec:all-facets}

The three-state inequality $\beta\le0$ is violated by the trine. One can ask
whether a larger number of hypotheses permits a classical facet that quantum
profiles always satisfy, but some GPT profile violates. The explicit facet
coordinates give a negative answer for every $n$.

\begin{theorem}[Quantum violations of all nonuniversal classical facets]
\label{thm:all-facets}
Every facet of $\Creg_n$ that admits a violation in $\Greg_n$ also admits a
violation in $\Qreg_n$, on a Hilbert space of dimension at most three.
More precisely, for $n\ge4$ the singleton facets in
\eqref{eq:classical-facets-all} are universal, and every facet indexed by a set
$B$ with $|B|\ge2$ is violated by a quantum profile. For $n=3$, the singleton
and pair facets are universal and the triple facet is quantum-violatable.
\end{theorem}

We first record how the signed coefficients behave when hypotheses are
repeated.

\begin{lemma}[Repetition of hypotheses]\label{lem:cloning}
Let $g:[n]\twoheadrightarrow[m]$ be a surjection and let $s_i=t_{g(i)}$.
If $D_0$ is the success profile of $(t_j)_{j=1}^m$, then
$D(H)=D_0(g(H))$. Moreover,
\begin{equation}\label{eq:cloning-coefficients}
 \widehat c_{g^{-1}(J)}(D)=\widehat c_J(D_0)
 \quad(\varnothing\ne J\subseteq[m]),
\end{equation}
and every coefficient not indexed by a union of fibers of $g$ is zero.
\end{lemma}
\begin{proof}
Effects assigned to identical states may be merged without changing the
success sum. Conversely, a measurement for the distinct states may be
assigned to one representative of each fiber, with zero effects on the other
labels. This proves the success identity. In the unique expansion
\eqref{eq:signed-coverage}, the condition $J\cap g(H)\ne\varnothing$ is exactly
$g^{-1}(J)\cap H\ne\varnothing$, proving the coefficient formula.
\end{proof}

\begin{proof}[Proof of Theorem~\ref{thm:all-facets}]
Singleton facets are universal by \eqref{eq:singleton-universal}. For any
$B\subseteq[n]$ with $|B|\ge3$, partition $B$ into three nonempty groups and
assign one of the three trine states to each group. If $B\ne[n]$, assign a
single orthogonal pure state to all labels outside $B$, using
$\mathbb C^2\oplus\mathbb C$. The success profile is the trine profile on the
groups that are represented, plus $1$ when an outside label is present.
By Lemma~\ref{lem:cloning} and uniqueness of the signed coefficients,
\begin{equation}\label{eq:trine-lift}
 \widehat c_B(D)=\widehat c_{123}(D_{\mathrm{trine}})=-\btr<0.
\end{equation}

To violate a pair facet requires four hypotheses. Consider the four qubit
states with Bloch vectors
\begin{equation}\label{eq:four-qubit-vectors}
 \vec r_1=(0,1,0),\quad \vec r_2=(0,-1,0),\quad
 \vec r_3=(a,0,0),\quad \vec r_4=(-a,0,0),\qquad 0<a<1.
\end{equation}
The dual calculation in \eqref{eq:qubit-geometry} applies to any number of qubit
states: $D(H)=1+R(H)$ for nonempty $H$, where $R(H)$ is their smallest
enclosing-ball radius. Here
\begin{align}
 D(1234)&=2, & D(34)&=1+a,\nonumber\\
 D(134)&=D(234)=1+\frac{1+a^2}{2}.\label{eq:four-qubit-success}
\end{align}
Indeed, the two antipodal unit vectors force the full radius to be $1$.
The triples $134$ and $234$ form acute triangles with circumcenters
$(0,\pm(1-a^2)/2,0)$ and radius $(1+a^2)/2$. It follows that
\begin{equation}\label{eq:negative-pair}
 \widehat c_{12}(D)
 =D(234)+D(134)-D(34)-D(1234)=a^2-a<0.
\end{equation}
Taking $a=1/2$ gives $\widehat c_{12}=-1/4$.
For a designated pair $B$ in any $n\ge4$, assign the first two states to its
two labels and partition the remaining labels into two nonempty groups,
assigned states three and four. Lemma~\ref{lem:cloning} preserves the negative
coefficient on $B$.

For $n=3$, Theorem~\ref{thm:C3} and the one-point gluing inequalities give the
asserted classification. For $n=2$, both the classical and GPT regions are
$0\le F(12)\le1$, so neither classical facet admits a GPT violation.
\end{proof}

The four-state example also gives
\begin{equation}\label{eq:submodularity-failure}
 D(134)+D(234)<D(1234)+D(34),
\end{equation}
so quantum success profiles need not be submodular. For three hypotheses,
submodularity follows from the universal inequalities in
Theorem~\ref{thm:G3}. Thus a new qualitative feature already appears at
$n=4$, although not the proposed separation by a quantum-valid classical
facet. Repeated hypotheses are a convenient construction device, not essential
to these strict violations: small perturbations preserve the sign.

Theorem~\ref{thm:all-facets} concerns the existence of a quantum violation of
\emph{each individual classical facet}. It does not identify their quantum
maxima, assert simultaneous violations, or equate $\Qreg_n$ with $\Greg_n$.
It distinguishes this hierarchy from multipartite Bell scenarios: the
``Guess Your Neighbor's Input'' inequalities include local facets whose
classical and quantum bounds coincide, but which no-signalling correlations
can violate \cite{AlmeidaEtAl2010}.

\section{Unequal priors can reveal hidden nonclassicality}\label{sec:weighted}

The discrete profile samples only the uniform prior on each subfamily. It need
not determine whether the same states have a classical description at all
priors. We give one explicit separation, without developing the general
weighted theory.

For $w\in\R_+^3$, define the homogeneous success and error functions
\begin{equation}\label{eq:weighted-def}
 \Sopt_\rho(w)=\max_{(M_i)_{i=1}^3}\sum_iw_i\Tr(\rho_iM_i),
 \qquad E_\rho(w)=\sum_iw_i-\Sopt_\rho(w).
\end{equation}
For $H\subseteq[3]$, let $w_H$ agree with $w$ on $H$ and vanish elsewhere,
without renormalization. Then $E_\rho(\1_H)=F(H)$. A weighted classical
profile means that a \emph{single} family of classical distributions reproduces
$\Sopt_\rho(w)$ for every $w$; the classical family cannot be changed with the
prior.

Define
\begin{equation}\label{eq:weighted-beta}
 B_w:=E_\rho(w)-E_\rho(w_{12})-E_\rho(w_{13})-E_\rho(w_{23}).
\end{equation}
For classical distributions $p_1,p_2,p_3$ on an alphabet $X$,
\begin{equation}\label{eq:weighted-classical}
 B_w=-\sum_{x\in X}\min\{w_1p_1(x),w_2p_2(x),w_3p_3(x)\}\le0.
\end{equation}
To see this, sort the three nonnegative numbers at each outcome. Their
three-way error is the sum of the two smaller numbers, whereas their three
binary errors add up to twice the smallest plus the middle number.
The same identity holds with an integral for general classical distributions.

\begin{proposition}[A classical discrete profile with a nonclassical weighted profile]
\label{prop:hidden}
There exist three density matrices on $\mathbb C^4$ whose complete discrete
profile belongs to $\Creg_3$, but for which $B_w>0$ for some positive $w$.
\end{proposition}
\begin{proof}
Let $\tau_1,\tau_2,\tau_3$ be the trine states from
\eqref{eq:trine-vectors}. On
$\mathbb C|c\rangle\oplus\mathbb C|f\rangle\oplus\mathbb C^2$, take
\begin{align}
 \rho_1&=\frac1{10}|c\rangle\langle c|
       +\frac{81}{100}|f\rangle\langle f|+\frac9{100}\tau_1,
       \label{eq:hidden-states}\\
 \rho_2&=\frac1{10}|c\rangle\langle c|+\frac9{10}\tau_2,\nonumber\\
 \rho_3&=\frac1{10}|c\rangle\langle c|+\frac9{10}\tau_3.\nonumber
\end{align}
The three mutually orthogonal blocks can be read first without loss of
optimality. The $|f\rangle$ block never makes an error, since only hypothesis
one occupies it. Hence
\begin{equation}\label{eq:hidden-error}
 E_\rho(w)=\frac1{10}\left(\sum_iw_i-\max_iw_i\right)
       +\frac9{10}E_\tau(w_1/10,w_2,w_3).
\end{equation}

The unequal-weight binary trine error is
\begin{equation}\label{eq:binary-trine-weighted}
 E_\tau(a,b,0)=\frac{a+b-\sqrt{a^2+ab+b^2}}2,
\end{equation}
with the same expression for the other pairs. This follows from the binary
Helstrom formula and $\Tr(\tau_i\tau_j)=1/4$.
For the weights $(1/10,1,1)$, the optimal trine success is $1+\sqrt3/2$:
one may ignore the first state and optimally discriminate the other two.
To check the dual bound explicitly, use
\[
 Y=\frac12\left((1+\sqrt3/2)\id-\frac12\sigma_x\right).
\]
It dominates $\tau_2$ and $\tau_3$. It also dominates $\tau_1/10$, because its
Bloch-vector difference has norm $3/5$, while its scalar difference is
$9/10+\sqrt3/2>3/5$. Thus this binary strategy is optimal for the three-way
weighted problem as well.

Writing $r=\sqrt{111}/10$ and $s=\sqrt3/2$, we obtain from
\eqref{eq:hidden-error}
\begin{align}
 x=y&=\frac1{10}+\frac9{20}\left(\frac{11}{10}-r\right),\nonumber\\
 z&=\frac1{10}+\frac9{10}(1-s),\label{eq:hidden-profile}\\
 t&=\frac15+\frac9{10}\left(\frac{11}{10}-s\right).\nonumber
\end{align}
Consequently,
\begin{equation}\label{eq:hidden-discrete-beta}
 \beta=\frac{9\sqrt{111}}{100}-1<0.
\end{equation}
Since this is a quantum and hence GPT profile, Theorem~\ref{thm:C3} proves
that its \emph{entire} discrete profile is classical.

In contrast, take $w=(10,1,1)$. In the trine block the weights are now equal.
The common classical block contributes $-1/10$ to $B_w$, and the trine block
contributes $9\btr/10$. Therefore
\begin{equation}\label{eq:hidden-weighted-beta}
 B_{(10,1,1)}=\frac9{10}\btr-\frac1{10}
 =\frac{27\sqrt3}{20}-\frac{19}{10}>0.
\end{equation}
For the normalized prior $p=(5/6,1/12,1/12)$ this is
\[
 B_p=\frac{27\sqrt3-38}{240}\simeq0.0365224.
\]
Equation~\eqref{eq:weighted-classical} rules out every classical model of the
weighted profile.
\end{proof}

The single unweighted witness $\beta$ is therefore no longer a complete test
when one asks whether the entire prior-dependent profile is classical. Here
unequal weights rebalance a trine component that is suppressed in the
uniform-prior data.

\section{Discussion and open questions}\label{sec:discussion}

The envelope theorem identifies all discrimination profiles compatible with
the general operational framework and makes their universal constraints
accessible to finite linear programming. The classical region consists of a
single normalized coverage function, while the GPT region permits finite upper
envelopes. A Cartesian product of simplexes realizes exactly this freedom.
The three-state case already produces a complete nonclassicality witness and
strict classical--quantum--GPT separation.

The first quantum question is to determine $\beta_{\mathrm Q}$. The trine is
optimal for qubits and all pure triples, and the unrestricted optimum lies
between $3\sqrt3/2-2$ and $(3\sqrt{15}-9)/4$. A proof of the conjectured
bound would give
an exact quantum value for the sole nonuniversal classical facet at $n=3$.
It would not, by itself, determine the entire quantum region $\Qreg_3$.

Larger numbers of hypotheses already exhibit qualitatively new behavior:
quantum success profiles need not be submodular, as the four-state example in
Section~\ref{sec:all-facets} shows. Nevertheless, a classical facet that is
valid for all quantum profiles but violated by a GPT cannot occur at any $n$:
Theorem~\ref{thm:all-facets} rules this out. This leaves substantial questions
about the quantum bounds of these facets, relations among their violations,
and genuinely new constraints defining the GPT region at higher $n$.
It also leaves open analogous questions for valid classical inequalities
that are not facets. The supporting linear program provides a concrete route
to investigating these higher-dimensional regions.

A separate issue is dimension. The universal construction is explicit but can
be large. Minimal-dimensional realization of a quantitative profile extends
the geometric questions surrounding perfect distinguishability and higher rank
antipodality \cite{Weiner2025,NaszodiSzilagyiWeiner2025}. It is natural to compare
general GPT realizations with Cartesian products, as well as fixed-dimensional
quantum regions with the unrestricted quantum region.

Finally, Proposition~\ref{prop:hidden} shows that a profile over all priors
contains genuinely more discrimination information than its subfamily samples.
An extension of the present hierarchy to prior-dependent functions should
therefore distinguish finite collections of priors from the full continuous
profile. This direction complements, rather than replaces, the finite
polyhedral problem studied here.

\section*{Acknowledgments}
The author thanks the participants in the weekly seminar of the quantum
information group at Budapest University of Technology and Economics for
questions, suggestions, and discussions that helped shape the results and
their presentation.

This work was supported by the National Research, Development and Innovation
Office of Hungary (NKFIH) via the grants ADVANCED\_25 152599, K~146380 and
EXCELLENCE~151342, and by the Ministry of Culture and Innovation and the NKFIH
within the Quantum Information National Laboratory of Hungary
(Grant No.~2022-2.1.1-NL-2022-00004).

\section*{Author contributions and use of AI tools}
The author formulated and directed the project, selected and developed its
mathematical direction, and takes responsibility for the mathematical claims
and the final manuscript. The research and preparation of this article
involved an extended series of conversations with ChatGPT, using the models
Sol 5.6 and Astra. Most of the mathematical material, including the detailed
proof arguments, was developed through these interactions. The tools also
assisted with the construction and analysis of examples, exploratory
calculations and code, literature searches, and drafting and revision of the
text. Their use was therefore not limited to language editing.

\appendix
\section{Rigorous evidence for the trine conjecture}\label{app:evidence}

\subsection{Pure states in arbitrary dimension}

\begin{theorem}[Pure-state optimum]\label{thm:pure}
For any three pure states in an arbitrary finite-dimensional Hilbert space,
\[
 \beta\le\frac{3\sqrt3}{2}-2.
\]
Equality holds only for a qubit trine embedded in that space, up to unitary
conjugation and permutation.
\end{theorem}
\begin{proof}
Write $\rho_i=|\psi_i\rangle\langle\psi_i|$ and let
$G=(\langle\psi_i,\psi_j\rangle)_{i,j=1}^3$ be the Gram matrix.
The pretty-good measurement on their span has success
$\sum_i(\sqrt G)_{ii}^2$, including the singular case by using the inverse on
the support. To see the formula, let $\Psi$ have columns $\psi_i$ and put
$S=\Psi\Psi^*$. The measurement vectors are $S^{-1/2}\psi_i$, and
$\Psi^*S^{-1/2}\Psi=\sqrt G$. Therefore
\begin{equation}\label{eq:pure-PGM}
 \Sopt\ge\sum_i(\sqrt G)_{ii}^2\ge\frac{(\Tr\sqrt G)^2}{3}.
\end{equation}

Let $\lambda_1,\lambda_2,\lambda_3$ be the eigenvalues of $G$; they are
nonnegative and sum to three. Set
\[
 q=\sum_{i<j}|G_{ij}|^2,
 \quad Q=\sum_{i<j}\lambda_i\lambda_j,
 \quad S_1=\sum_{i<j}\sqrt{\lambda_i\lambda_j}.
\]
Then $Q=3-q$ and concavity gives
\[
 T=\sum_{i<j}\sqrt{1-|G_{ij}|^2}\le\sqrt{3Q},\qquad
 \Sopt\ge1+\frac23S_1.
\]
Consequently,
\begin{equation}\label{eq:pure-reduction}
 \beta\le\sqrt{3Q}-1-\frac23S_1.
\end{equation}

For completeness, the required scalar inequality is proved next. Put
$x_i=\sqrt{\lambda_i}$, so $\sum_i x_i^2=3$, and write
$p=x_1+x_2+x_3$, $r=x_1x_2x_3$. Then
\[
 0\le S_1\le3,\quad p^2=3+2S_1,\quad Q=S_1^2-2pr.
\]
If $S_1\le3/2$, then $Q\le S_1^2$, and
\[
 \sqrt{3Q}-\frac23S_1
 \le\left(\sqrt3-\frac23\right)S_1
 \le\frac{3\sqrt3}{2}-1.
\]
If $S_1\ge3/2$, Schur's inequality $p^3+9r\ge4pS_1$ gives
$r\ge p(2S_1-3)/9$. Using $S_1\le3$, we obtain
\[
 4pr\ge\frac{4(3+2S_1)(2S_1-3)}9
 \ge(S_1+1)(2S_1-3).
\]
Thus $2Q\le S_1+3$, and
\[
 \sqrt{3Q}-\frac23S_1
 \le\sqrt{\frac32S_1+\frac92}-\frac23S_1.
\]
The last function is strictly decreasing for $3/2\le S_1\le3$, so its maximum
is again $3\sqrt3/2-1$. This proves the bound in
\eqref{eq:pure-reduction}.

Equality forces $S_1=3/2$, $Q=S_1^2$, and hence $r=0$. Together with
$\sum_i\lambda_i=3$, this gives the spectrum $(3/2,3/2,0)$ up to permutation.
Equality in the concavity estimate forces $|G_{ij}|=1/2$ for $i\ne j$.
The determinant equation then gives
$G_{12}G_{23}G_{31}=-1/8$. Changing the phases of the vectors makes all three
off-diagonal entries $-1/2$. This is the Gram matrix of a planar trine,
proving the equality characterization.
\end{proof}

\subsection{Removing a common positive component}

\begin{proposition}[Common positive semidefinite mass]\label{prop:common-mass}
Suppose $0\le B\le\rho_i$ for $i=1,2,3$, and let $a=\Tr B<1$. Define
$\sigma_i=(\rho_i-B)/(1-a)$. Then
\begin{equation}\label{eq:common-mass}
 \beta(\rho_1,\rho_2,\rho_3)
 =(1-a)\beta(\sigma_1,\sigma_2,\sigma_3)-a.
\end{equation}
Removing the common component cannot decrease $\beta$, and strictly increases
it when $a>0$ and $\beta(\rho_1,\rho_2,\rho_3)>-1$.
\end{proposition}
\begin{proof}
Pairwise differences scale by $1-a$, so $T(\rho)=(1-a)T(\sigma)$.
For every POVM,
\[
 \sum_i\Tr(\rho_iM_i)=\Tr B+(1-a)\sum_i\Tr(\sigma_iM_i),
\]
hence $\Sopt(\rho)=a+(1-a)\Sopt(\sigma)$. This proves
\eqref{eq:common-mass}. Since $\beta\ge-1$ throughout $\Greg_3$,
\[
 \beta(\sigma)-\beta(\rho)
 =\frac{a}{1-a}\bigl(\beta(\rho)+1\bigr)\ge0.
\]
\end{proof}

The set $\{B:0\le B\le\rho_i\text{ for all }i\}$ is compact. Choose $B$
with maximal trace. Unless the original states are identical, $a<1$, and the
residual states have
\begin{equation}\label{eq:empty-common-support}
 \operatorname{im}\sigma_1\cap\operatorname{im}\sigma_2
 \cap\operatorname{im}\sigma_3=\{0\}.
\end{equation}
Indeed, a nonzero vector in this intersection would allow a sufficiently small
positive multiple of its rank-one projector to be subtracted from every
residual state, contradicting maximality. Thus a search for a value above the
trine bound may be restricted to triples satisfying
\eqref{eq:empty-common-support}. This is a support reduction, not a reduction
to pure states.

\subsection{Mixed triples attaining the trine value}

For a fixed ancillary state $\omega$, discrimination of
$\omega\otimes\rho_i$ has exactly the same optimum as discrimination of
$\rho_i$: the ancilla can be ignored, and every joint POVM induces a POVM on
the original system by taking its expectation in $\omega$.
Also $\|\omega\otimes(\rho_i-\rho_j)\|_1=\|\rho_i-\rho_j\|_1$.
Hence
\begin{equation}\label{eq:ancilla}
 \beta(\omega\otimes\rho_1,\omega\otimes\rho_2,\omega\otimes\rho_3)
 =\beta(\rho_1,\rho_2,\rho_3).
\end{equation}
Taking a mixed $\omega$ and a trine produces mixed triples attaining $\btr$.
Their common support intersection is still zero.

More generally, for block-diagonal states
$\rho_i=\bigoplus_k p_k\rho_i^{(k)}$ with weights independent of $i$, both
trace norms and optimal successes add across blocks. Therefore
\begin{equation}\label{eq:block-beta}
 \beta(\rho_1,\rho_2,\rho_3)
 =\sum_kp_k\beta(\rho_1^{(k)},\rho_2^{(k)},\rho_3^{(k)}).
\end{equation}
In particular, arbitrary flagged mixtures of trine triples attain the same
value. These constructions explain why pure-state optimality alone does not
settle the mixed-state conjecture, and why a classification of all equality
cases would have to include higher-rank examples.

\end{document}